\documentclass{amsart}

\usepackage{amssymb}
\usepackage{amsmath}

\usepackage[utf8]{inputenc}
\usepackage[T1,T2A]{fontenc}

\newtheorem{theorem}{Theorem}
\newtheorem{corollary}{Corollary}
\numberwithin{equation}{section}

\usepackage{array}
\newcolumntype{C}{>{$\displaystyle}c<{$}}
\newcolumntype{L}{>{$\displaystyle}l<{$}}

\title[On Infinite Series of Bessel functions of the First Kind]{On Infinite Series of Bessel functions \\of the First Kind: \\ $\sum_\nu J_{N\nu+p}(x), \sum_\nu (-1)^\nu J_{N\nu+p}(x)$}
\author{Suk Hyun Sung}
\address{Department of Materials Science and Engineering,
University of Michigan, Ann Arbor MI 48105, United States}
\email{sukhsung@umich.edu}
\author{Robert Hovden}
\address{Department of Materials Science and Engineering,
University of Michigan, Ann Arbor MI 48105, United States}
\email{hovden@umich.edu}

\date{October 26, 2022}

\newcommand{\pcos}[1]{\cos{(#1)}}
\newcommand{\psin}[1]{\sin{(#1)}}

\newcommand{\pdelta}[1]{\delta{(#1)}}
\newcommand{\infsum}[1]{\sum_{#1=-\infty}^{\infty}}

\begin{document}

\begin{abstract}
    Infinite series of Bessel function of the first kind, $\sum_\nu^{\pm\infty} J_{N\nu+p}(x)$, $\sum_\nu^{\pm\infty} (-1)^\nu J_{N\nu+p}(x)$, are summed in closed form. These expressions are evaluated by engineering a Dirac comb that selects specific sequences within the Bessel series.
\end{abstract}
\maketitle

\section*{Introduction}
Infinite series of Bessel functions of the first kind in the form $\sum_{2\nu}J_{2\nu}(x)$ and $\sum_{3\nu}J_{3\nu}(x)$ arise in many natural systems. They are of particular interest in condensed matter physics when crystals spontaneously break symmetry~\cite{peierls1979surprises} due to correlated electron effects such as superconductivity, charge density waves, collosal magnetoresistance, and quantum spin liquids. Mathematically, these series can appear when sinusoids exist inside complex exponentials---as described by the Jacobi-Anger relation~\cite{Jacobi_1836, Anger_1855}. Early treatises on Bessel functions by Neumann and Watson~\cite{Neumann_1867,WhittakerWatson, Watson} provide analytic solutions to the alternating series $\sum_{\nu}(-1)^\nu J_{2\nu}(x)$ and $\sum_{\nu}(-1)^\nu J_{2\nu+1}(x)$ which are commonly tabulated~\cite{JeffreyZwillinger,AbramowitzStegun}. However an analytic expressions for $\sum_{\nu}J_{3\nu}(x)$, $\sum_{\nu}J_{3\nu\pm1}(x)$ or more general series $\sum_{\nu} J_{N\nu+p}(x)$ and $\sum_{\nu} (-1)^\nu J_{N\nu+p}(x)$ are not readily available. We show closed form expressions to infinite series of Bessel functions of the first kind exist. The expression is evaluated by engineering a Dirac comb that selects specific sequences within the Bessel series.

To illustrate, we find a closed form expression to the series:
\begin{align*}
    \infsum{\nu} J_{3\nu+p}(x) = \frac{1}{3}\left[1+2\pcos{\frac{x\sqrt{3}}{2}-\frac{2\pi p}{3}}\right]
\end{align*}
{\hfill$\nu, p \in \mathbb{Z};\; x \in \mathbb{C}$}\\

More generally, we find an expression to all series in the class:
\begin{align*}
    \infsum{\nu} J_{N\nu+p}(x) &= \frac{1}{N}\sum_{q=0}^{N-1} e^{ix\psin{2\pi q/N}}e^{-i2\pi pq/N }\\
    \infsum{\nu} (-1)^\nu J_{N\nu+p}(x) &= \frac{1}{N} \sum_{q=0}^{N-1}e^{ix\psin{(2q+1)\pi/N}}e^{-i(2q+1)\pi p/N}
\end{align*}
{\hfill$\nu, p, q \in \mathbb{Z};\; N \in \mathbb{Z}^+;\; x \in \mathbb{C}$}\\

From these theorem's, we tabulate a family of closed analytic forms to infinite series of Bessel functions of the first kind.

\vspace{5mm}
\section{Evaluation of Series: $\displaystyle\infsum{\nu} J_{N\nu+p}(x)$}
    
\begin{theorem}
Infinite series of Bessel functions of the first kind of the form $\sum_\nu^{\pm\infty} J_{N\nu+p}(x)$ where $\nu, p, q \in \mathbb{Z}, \ N \in \mathbb{Z}^+$ and $x\in\mathbb{C}$ have following closed expression:
\begin{align}
    \infsum{\nu} J_{N\nu+p}(x) = \frac{1}{N}\sum_{q=0}^{N-1} e^{-i2\pi pq/N }e^{ix\psin{2\pi q/N}}
\end{align}
\end{theorem}

\begin{proof}
Consider the following series of evenly spaced delta functions (i.e. a Dirac comb) on an infinite series of Bessel functions:\begin{align}
    f(k) = \infsum{h} \sum_{p=0}^{N-1} \delta(k-(h+\frac{p}{N})a) \infsum{\nu} J_{N\nu+p}(kA) \label{eq:f_initial}
\end{align}
where $\nu, h, p \in \mathbb{Z}$; $N \in \mathbb{Z}^+$; $k, a\in\mathbb{R}$; $A\in\mathbb{C}$. 

Three summations are re-grouped into two:
\begin{align}
    f(k) &= \infsum{h} \infsum{\alpha} \pdelta{k-(h+\frac{\alpha}{N})a} J_{\alpha}(kA)
\end{align}
The Dirac comb can be represented as a Fourier series:
\begin{align}
    \infsum{h} \delta(k-ha) = \frac{1}{a} \infsum{m} e^{-i2\pi k m/a}\label{eq:diracComb}
\end{align}
With (\ref{eq:diracComb}), $f(k)$ becomes:
\begin{align}
    f(k) = \frac{1}{a}\infsum{m} e^{-i2\pi k m /a} \infsum{\alpha} J_{\alpha}(kA)e^{i\alpha 2\pi m/N }
\end{align}
The summation over $\alpha$ appears in the Jacobi-Anger relation~\cite{Jacobi_1836, Anger_1855, Watson}:
\begin{align}
    e^{ix\sin(\theta)}=\infsum{\alpha} J_{\alpha}(x)e^{i\alpha\theta}\label{eq:jacobianger}
\end{align}
Using the Jacobi-Anger relation:
\begin{align}
    f(k) = \frac{1}{a}\infsum{m} e^{-i2\pi km/a} e^{ikA \sin(2\pi m/N)}
\end{align}
We split the summation over $m$ into $m = ..., Nn, Nn+1,..., Nn+(N-1), ...$ :
\begin{align}
    f(k) &= \frac{1}{a} \infsum{n} \bigg[... + e^{-i2\pi k nN/a}\nonumber\\
    &\;\qquad\qquad\qquad+e^{-i2\pi knN/a}e^{-i2\pi k/a}e^{ikA\psin{2\pi/N}}+...\nonumber\\
    &\;\qquad\qquad\qquad+e^{-i2\pi knN/a}e^{-i2\pi k(N-1)/a}e^{ikA\psin{2\pi(N-1)/N}}+... \bigg]\nonumber\\
    &=\frac{1}{a} \infsum{n}e^{-i2\pi knN/a} \sum_{q=1}^{N-1} e^{-i2\pi kq/a}e^{ikA\psin{2\pi q/N}} \label{eq:sumsplit}
\end{align}
Using (\ref{eq:diracComb}) again:
\begin{align}
    f(k) &= \frac{1}{N}\infsum{l} \pdelta{k-\frac{la}{N}}\sum_{q=0}^{1-N} e^{ikA\psin{2\pi q/N}}e^{-i2\pi kq/a}
\end{align}
Substitute $k = \frac{la}{N}; l\in\mathbb{Z}$, as $f$ is non-zero only where a delta function exists.
\begin{align}
    f(k) &= \frac{1}{N}\infsum{l} \pdelta{k-\frac{la}{N}}\sum_{q=0}^{N-1} e^{ikA\psin{2\pi q/N}}e^{-i2\pi ql/N }
\end{align}
We split the summation into $l=..., Nh, Nh+1, ..., (N+1)h-1 ...$ then regroup, similar with (\ref{eq:sumsplit}):
\begin{align}
    f(k) = \frac{1}{N}\infsum{h} \sum_{p=0}^{N-1} \pdelta{k-(h+\frac{p}{N})a}\sum_{q=0}^{N-1} e^{ikA\psin{2\pi q/N}}e^{-i2\pi pq/N} \label{eq:f_final}
\end{align}

Initial expression (\ref{eq:f_initial}) must equal (\ref{eq:f_final}):
\begin{align}
    \infsum{h}\sum_{p=0}^{N-1} &\delta(k-(h+\frac{p}{N})a)\; \infsum{\nu} J_{N\nu+p}(kA) \nonumber\\ &
    = \frac{1}{N}\infsum{h} \sum_{p=0}^{N-1} \pdelta{k-(h+\frac{p}{N})a}\sum_{q=0}^{N-1} e^{ikA\psin{2\pi q/N}}e^{-i2\pi pq/N}
\end{align}

This relation suggests only values on the Dirac comb are equivalent. However, the variable $a$ can take on any real value thus the expression holds for all values of $kA$. The equivalent Dirac lattices on each side can be disregarded.
\begin{align}
    \therefore\quad\infsum{\nu} J_{N\nu+p}(x) &= \frac{1}{N}\sum_{q=0}^{N-1} e^{ix\psin{2\pi q/N}}e^{-i2\pi pq/N}
\end{align}
\end{proof}

\pagebreak
\begin{corollary}
From Theorem 1 we comprise a table of closed form expressions:
\begin{center}
\footnotesize
\begin{tabular}{|| C | C | L | L c ||}
    \hline
    N & p    & \infsum{\nu} J_{N\nu+p}(x) & \frac{1}{N}\sum_{q=0}^{N-1} e^{ix\psin{2\pi q/N}}e^{-i2\pi pq/N}&\\
    \hline
    \hline
    1 & 0 & \infsum{\nu} J_\nu(x) & 1 & \cite{Watson,JeffreyZwillinger,kuzmin1933}\\
    \hline
    2 & 0 & \infsum{\nu} J_{2\nu}(x) & 1 &  \cite{Watson,JeffreyZwillinger}\\
      & 1 & \infsum{\nu} J_{2\nu+1}(x) & 0 &  \cite{Watson,JeffreyZwillinger}\\
    \hline
    3 & 0  & \infsum{\nu} J_{3\nu}   & \frac{1}{3} \bigg[1+2\pcos{\frac{x\sqrt{3}}{2}}\bigg] &\\
      & 1 & \infsum{\nu} J_{3\nu+1} & \frac{1}{3}\bigg[1+2\pcos{\frac{x\sqrt{3}}{2}-\frac{2\pi}{3}}\bigg] &\\ 
      & 2 & \infsum{\nu} J_{3\nu+2} & \frac{1}{3}\bigg[1+2\pcos{\frac{x\sqrt{3}}{2}-\frac{4\pi}{3}}\bigg] &\\ 
     \hline
    4 & 0 & \infsum{\nu} J_{4\nu}   & \cos^2(\frac{x}{2}) &\\
      & 1 & \infsum{\nu} J_{4\nu+1} & \frac{1}{2}\sin(x) &\\ 
      & 2 & \infsum{\nu} J_{4\nu+2} & \sin^2(\frac{x}{2}) &\\ 
      & 3 & \infsum{\nu} J_{4\nu+3} & -\frac{1}{2}\sin(x) &\\ 
     \hline
    5 & 0 & \infsum{\nu} J_{5\nu}    & \frac{1}{5}\bigg[1+2\pcos{x\psin{\frac{2\pi}{5}}}+2\pcos{x\psin{\frac{4\pi}{5}}}\bigg] &\\
      & 1  & \infsum{\nu} J_{5\nu+1} & \frac{1}{5} \bigg[1+2\pcos{x\psin{\frac{2\pi}{5}}-\frac{2\pi}{5}}+2\pcos{x\psin{\frac{4\pi}{5}}-\frac{4\pi}{5}}\bigg] &\\
      & 2  & \infsum{\nu} J_{5\nu+2} & \frac{1}{5} \bigg[1+2\pcos{x\psin{\frac{2\pi}{5}}-\frac{4\pi}{5}}+2\pcos{x\psin{\frac{4\pi}{5}}-\frac{8\pi}{5}}\bigg] &\\
      & 3  & \infsum{\nu} J_{5\nu+3} & \frac{1}{5} \bigg[1+2\pcos{x\psin{\frac{2\pi}{5}}-\frac{6\pi}{5}}+2\pcos{x\psin{\frac{4\pi}{5}}-\frac{12\pi}{5}}\bigg] &\\
      & 4  & \infsum{\nu} J_{5\nu+4} & \frac{1}{5} \bigg[1+2\pcos{x\psin{\frac{2\pi}{5}}-\frac{8\pi}{5}}+2\pcos{x\psin{\frac{4\pi}{5}}-\frac{16\pi}{5}}\bigg] &\\
    \hline
    6 & 0 & \infsum{\nu} J_{6\nu}    & \frac{1}{3}\bigg[1+2\pcos{\frac{x\sqrt{3}}{2}}\bigg] &\\
      & 1  & \infsum{\nu} J_{6\nu+1} & \frac{1}{\sqrt{3}}\psin{\frac{\sqrt{3}}{2}} &\\
      & 2  & \infsum{\nu} J_{6\nu+2} & \frac{1}{3} \bigg[1-\pcos{\frac{x\sqrt{3}}{2}}\bigg] &\\
      & 3  & \infsum{\nu} J_{6\nu+3} & 0 &\\
      & 4  & \infsum{\nu} J_{6\nu+4} & \frac{1}{3} \bigg[1-\pcos{\frac{x\sqrt{3}}{2}}\bigg] &\\
      & 5  & \infsum{\nu} J_{6\nu+5} & -\frac{1}{\sqrt{3}}\psin{\frac{\sqrt{3}}{2}} &\\
    \hline
\end{tabular}
\end{center}
\end{corollary}

\pagebreak
\section{Evaluation of Alternating Series: $\displaystyle \infsum{\nu} (-1)^\nu J_{N\nu+p}(x)$}
\begin{theorem}
Infinite series of Bessel functions of the first kind of the form $\sum_\nu (-1)^\nu J_{N\nu+p}(x)$ where $\nu, p, q \in \mathbb{Z}, \ N \in \mathbb{Z}^+$ and $x\in\mathbb{C}$ have following closed expression:
\begin{align}
    \infsum{\nu} (-1)^\nu J_{N\nu+p}(x) =\frac{1}{N} \sum_{q=0}^{N-1}e^{ix\psin{\frac{(2q+1)\pi} {N}}}e^{-i\frac{(2q+1)\pi p}{N}}
\end{align}
\end{theorem}

\begin{proof}
Consider the following Dirac comb on an infinite series of Bessel functions:
\begin{align}
    g(k) = \infsum{h}\sum_{p=0}^{N-1} \delta(k-(h+\frac{p}{N})a) \infsum{\nu} (-1)^\nu J_{N\nu+p}(kA) \label{eq:g_initial}
\end{align}
where $\nu, h, p \in \mathbb{Z}$; $N \in \mathbb{Z}^+$; $k, a\in\mathbb{R}$; $A\in\mathbb{C}$. 

Three summations are re-grouped into two:
\begin{align}
    g(k) &= \infsum{h}\infsum{\alpha} \pdelta{k-(h+\frac{\alpha}{N})a} (-1)^h J_{\alpha}(kA) \nonumber\\
    &= \infsum{h}\infsum{\alpha} \pdelta{k-(h+\frac{\alpha}{N})a} e^{-i(k-\frac{\alpha}{N}a)\frac{\pi}{a}} J_{\alpha}(kA) \nonumber\\
    &= \infsum{h}\infsum{\alpha} \pdelta{k-(h+\frac{\alpha}{N})a} e^{-i\pi k/a} e^{i\pi\alpha/N} J_{\alpha}(kA) 
\end{align}
The Dirac comb can be represented as a Fourier series (\ref{eq:diracComb}):
\begin{align}
    g(k) &= \frac{1}{a}\infsum{m}\infsum{\alpha} e^{-i2\pi (k-\frac{\alpha}{N}a)m/a} e^{-i\pi k/a} e^{i\pi\alpha/N} J_{\alpha}(kA) \nonumber\\
    &= \frac{1}{a}\infsum{m} e^{-i\pi k(2m+1)/a} \infsum{\alpha} e^{i\alpha \pi(2m+1)/N} J_{\alpha}(kA)
\end{align}
Using the Jacobi-Anger relation (\ref{eq:jacobianger}):
\begin{align}
    g(k) &= \frac{1}{a}\infsum{m} e^{-i \pi k(2m+1)/a}e^{ikA \psin{\pi (2m+1)/N}}
\end{align}
We split the summation over m into $m = ..., Nn, Nn+1, ..., Nn+(N-1), ...$ :
\begin{align}
    g(k) &=\frac{1}{N}\infsum{n} \bigg[...+e^{-i2\pi k nN/a}e^{-i\pi k/a}e^{ikA\psin{\frac{1\pi}{N}}}\nonumber\\
    &\qquad\qquad\; + e^{-i2\pi knN/a}e^{-i3\pi k/a}e^{ikA\psin{\frac{3\pi}{N}}}+...\nonumber\\
    &\qquad\qquad\; + e^{-i2\pi knN/a}e^{-i(2N-1)\pi k/a}e^{ikA\psin{\frac{(2N-1)\pi}{N}}}+... \bigg]\nonumber\\
    &=\frac{1}{a}\infsum{n} e^{-i2\pi knN/a}\sum_{q=0}^{N-1} e^{-i(2q+1)\pi k/a} e^{ikA\psin{\frac{(2q+1)\pi} {N}}}
\end{align}
Using (\ref{eq:diracComb}) again:
\begin{align}
    g(k) &= \frac{1}{N}\infsum{l} \pdelta{k-\frac{a}{N}l}\sum_{q=0}^{N-1}  e^{ikA\psin{\frac{(2q+1)\pi} {N}}}e^{-i(2q+1)\pi k/a}
\end{align}
Substitute $k = l\frac{a}{N}$, as $f$ is non-zero where a delta function exists.
\begin{align}
    g(k) &= \frac{1}{N}\infsum{l} \pdelta{k-\frac{a}{N}l}\sum_{q=0}^{N-1} e^{ikA\psin{\frac{(2q+1)\pi} {N}}}e^{-i\frac{(2q+1)\pi l}{N}}
\end{align}
We split the summation into $l=..., Nh, Nh+1, ..., Nh+(N-1),...$ then regroup:
\begin{align}
    g(k) = \frac{1}{N}\infsum{h} \sum_{p=0}^{N-1} \pdelta{k-(h+\frac{p}{N})a} (-1)^h \sum_{q=0}^{N-1}e^{ikA\psin{\frac{(2q+1)\pi} {N}}}e^{-i\frac{(2q+1)\pi p}{N}} \label{eq:g_final}
\end{align}

Initial expression (\ref{eq:g_initial}) must equal (\ref{eq:g_final}):
\begin{align}
     \infsum{h}\sum_{p=0}^{N-1} &\delta(k-(h+\frac{p}{N})a) \infsum{\nu} (-1)^\nu J_{N\nu+p}(kA) = \nonumber\\ 
     &\frac{1}{N}\infsum{h} \sum_{p=0}^{N-1} \pdelta{k-(h+\frac{p}{N})a} (-1)^h \sum_{q=0}^{N-1}e^{ikA\psin{\frac{(2q+1)\pi} {N}}}e^{-i\frac{(2q+1)\pi p}{N}}
\end{align}

The variable $a$ can take on any real value thus the expression holds for all values of $kA$. The equivalent Dirac lattices on each side can be disregarded.
\begin{align}
    \therefore\quad\infsum{\nu} (-1)^\nu J_{N\nu+p}(x) = \frac{1}{N} \sum_{q=0}^{N-1}e^{ix\psin{\frac{(2q+1)\pi} {N}}}e^{-i\frac{(2q+1)\pi p}{N}}
\end{align}
\end{proof}

\pagebreak
\begin{corollary}

From Theorem 2 we comprise a table of closed form expressions:
\begin{center}
\footnotesize
\begin{tabular}{|| C | C | L | L c ||}
    \hline
    N & p & \infsum{\nu} (-1)^\nu J_{N\nu+p}(x) & \frac{1}{N} \sum_{q=0}^{N-1}e^{ix\psin{\frac{(2q+1)\pi} {N}}}e^{-i\frac{(2q+1)\pi p}{N}}&\\
    \hline
    \hline
    1 & 0 & \infsum{\nu} (-1)^\nu J_{\nu}(x)    & 1 & \cite{Watson,JeffreyZwillinger} \\
    \hline
    2 & 0 & \infsum{\nu} (-1)^\nu J_{2\nu}(x)   & \pcos{x} &\cite{Watson}\\
      & 1 & \infsum{\nu} (-1)^\nu J_{2\nu+1}(x) & \psin{x} &\cite{Watson}\\
    \hline
    3 & 0 & \infsum{\nu} (-1)^\nu J_{3\nu}(x)   &  \;\;\, \frac{1}{3}\bigg[1+2\pcos{\frac{x\sqrt{3}}{2}}\bigg] &\\
      & 1 & \infsum{\nu} (-1)^\nu J_{3\nu+1}(x) &   -\frac{1}{3}\bigg[1-2\pcos{\frac{x\sqrt{3}}{2}-\frac{\pi}{3}}\bigg] &\\
      & 2 & \infsum{\nu} (-1)^\nu J_{3\nu+2}(x) &  \;\;\, \frac{1}{3}\bigg[1+2\pcos{\frac{x\sqrt{3}}{2}-\frac{2\pi}{3}}\bigg] &\\
    \hline
    4 & 0 & \infsum{\nu} (-1)^\nu J_{4\nu}(x)   & \pcos{\frac{x}{\sqrt{2}}}&\\
      & 1 & \infsum{\nu} (-1)^\nu J_{4\nu+1}(x) & \frac{1}{\sqrt{2}}\psin{\frac{x}{\sqrt{2}}} &\\
      & 2 & \infsum{\nu} (-1)^\nu J_{4\nu+2}(x) & 0 &\\
      & 3 & \infsum{\nu} (-1)^\nu J_{4\nu+3}(x) & \frac{1}{\sqrt{2}}\psin{\frac{x}{\sqrt{2}}} &\\
    \hline
    5 & 0 & \infsum{\nu} (-1)^\nu J_{5\nu}(x)   &  \;\;\, \frac{1}{5}\bigg[1+2\pcos{x\psin{\frac{\pi}{5}}}+2\pcos{x\psin{\frac{3\pi}{5}}}\bigg] &\\
      & 1 & \infsum{\nu} (-1)^\nu J_{5\nu+1}(x) &  -\frac{1}{5}\bigg[1+2\pcos{x\psin{\frac{\pi}{5}}+\frac{4\pi}{5}}+2\pcos{x\psin{\frac{3\pi}{5}}+\frac{2\pi}{5}}\bigg] &\\
      & 2 & \infsum{\nu} (-1)^\nu J_{5\nu+2}(x) & \;\;\, \frac{1}{5}\bigg[1+2\pcos{x\psin{\frac{\pi}{5}}+\frac{8\pi}{5}}+2\pcos{x\psin{\frac{3\pi}{5}}+\frac{4\pi}{5}}\bigg] &\\
      & 3 & \infsum{\nu} (-1)^\nu J_{5\nu+3}(x) &  -\frac{1}{5}\bigg[1+2\pcos{x\psin{\frac{\pi}{5}}+\frac{12\pi}{5}}+2\pcos{x\psin{\frac{3\pi}{5}}+\frac{6\pi}{5}}\bigg] &\\
      & 4 & \infsum{\nu} (-1)^\nu J_{5\nu+4}(x) & \;\;\, \frac{1}{5}\bigg[1+2\pcos{x\psin{\frac{\pi}{5}}+\frac{16\pi}{5}}+2\pcos{x\psin{\frac{3\pi}{5}}+\frac{8\pi}{5}}\bigg] &\\
    \hline
    6 & 0 & \infsum{\nu} (-1)^\nu J_{6\nu}(x)   & \;\;\,\frac{1}{3}\bigg[\pcos{x}+2\pcos{\frac{x}{2}}\bigg] &\\
      & 1 & \infsum{\nu} (-1)^\nu J_{6\nu+1}(x) & \;\;\,\frac{1}{3}\bigg[\psin{x}+\psin{\frac{x}{2}}\bigg] &\\
      & 2 & \infsum{\nu} (-1)^\nu J_{6\nu+2}(x) & -\frac{1}{3}\bigg[\pcos{x}-\pcos{\frac{x}{2}}\bigg] &\\
      & 3 & \infsum{\nu} (-1)^\nu J_{6\nu+3}(x) & -\frac{1}{3}\bigg[\psin{x}-2\psin{\frac{x}{2}}\bigg] &\\
      & 4 & \infsum{\nu} (-1)^\nu J_{6\nu+4}(x) & \;\;\,\frac{1}{3}\bigg[\pcos{x}-\pcos{\frac{x}{2}}\bigg] &\\
      & 5 & \infsum{\nu} (-1)^\nu J_{6\nu+5}(x) & \;\;\,\frac{1}{3}\bigg[\psin{x}+\psin{\frac{x}{2}}\bigg] &\\
    \hline
\end{tabular}
\end{center}
\end{corollary}

In Corollary 1, 2  we have noted references to previous reported expressions for $N = 1, 2$ that are closely related. For $N = 1$ the expression is given using the generating function $e^{1/2(t-1/t)x} = \infsum{\nu} J_{\nu}(x) t^\nu$ when $t = 1$. For $N=2$, the expressions can be re-written such that $ \infsum{\nu} (-1)^\nu J_{2\nu}(x) = J_0(x) + 2\sum_{\nu=1}^\infty (-1)^\nu J_{2\nu}(x) = \cos{x}$ and
      $\infsum{\nu} (-1)^\nu J_{2\nu+1}(x) = \sum_{\nu=0}^\infty (-1)^\nu J_{2\nu+1}(x) = \sin{x}$ using the identity $J_{-\nu}(x) = (-1)^\nu J_{\nu}(x)$.


\providecommand{\bysame}{\leavevmode\hbox to3em{\hrulefill}\thinspace}
\providecommand{\MR}{\relax\ifhmode\unskip\space\fi MR }
\providecommand{\MRhref}[2]{%
  \href{http://www.ams.org/mathscinet-getitem?mr=#1}{#2}
}
\providecommand{\href}[2]{#2}

\end{document}